\documentclass[11pt]{article}

\usepackage[margin=1.1in]{geometry}
\usepackage{amsmath,amssymb,amsthm}
\usepackage{booktabs}
\usepackage{microtype}
\usepackage{xcolor}
\usepackage[colorlinks=true,linkcolor=blue!50!black,citecolor=blue!50!black,urlcolor=blue!50!black]{hyperref}

\newtheorem{theorem}{Theorem}[section]
\newtheorem{lemma}[theorem]{Lemma}
\newtheorem{corollary}[theorem]{Corollary}
\newtheorem{proposition}[theorem]{Proposition}
\theoremstyle{definition}
\newtheorem{definition}[theorem]{Definition}
\theoremstyle{remark}
\newtheorem{remark}[theorem]{Remark}

\newcommand{\R}{\mathbb{R}}
\newcommand{\vol}{\operatorname{vol}}
\newcommand{\SX}{S_X}
\newcommand{\SY}{S_Y}
\newcommand{\SZ}{S_Z}
\newcommand{\K}{K}
\newcommand{\Stub}{\mathrm{Stub}}
\newcommand{\xhat}{\hat{x}}
\newcommand{\yhat}{\hat{y}}
\newcommand{\zhat}{\hat{z}}
\newcommand{\Vstar}{V^{*}}

\title{An upper bound for the three-dimensional moving sofa problem\\
in a snaky corridor}
\author{Yuzhe Tang\\[2pt]
\small Georgia Institute of Technology\\[1pt]
\small \texttt{ytang352@gatech.edu}}
\date{August 30, 2026 \\[2pt] \small (working note; constants machine-verified, see \S\ref{sec:discussion})}

\begin{document}
\maketitle

\begin{abstract}
Let $K_M\subset\R^3$ be the ``snaky'' corridor with unit-square cross-section:
an infinite entry arm parallel to the $x$-axis, a middle arm parallel to the
$y$-axis whose clear straight section has length $M$, and an infinite exit
arm parallel to the $z$-axis,
with the two turns lying in the $x$--$y$ and $y$--$z$ planes.  Let
$\Vstar(K_M)$ be the supremum of the volumes of connected bodies that can be
moved continuously (by rigid motions) from the entry arm to the exit arm while
remaining inside $K_M$.  We prove
$\Vstar(K_M)\le 7.8691$ for every $M\ge 3.53$ (Theorem~\ref{thm:A}); by a
sharper analysis of the same mechanism,
$\Vstar(K_M)\le 4.2862$ for every $M\ge 2.61$ (Theorem~\ref{thm:B}); and, by
replacing the final analytic caps with a certified branch-and-bound
computation over the compact family of entry tubes compatible with the forced
events, $\Vstar(K_M)\le 3.4083$ for every $M>2.2$ (Theorem~\ref{thm:C};
computer-assisted, in double precision with conservative slack).  In
particular all three bounds apply to the corridor with $M=4$ used in
Problem~6.64 of Georgiev, G\'omez-Serrano, Tao and Wagner
(arXiv:2511.02864), and uniformly to every ``turns far apart'' reading of that
problem.  To the best of our knowledge, this is the first upper bound for the
snaky-corridor 3D moving sofa problem (Problem~6.64 of arXiv:2511.02864);
previously not even finiteness of $\Vstar$ had been recorded.  This note
proves upper bounds only; combined with the certified lower bound
$1.8558404$ from the author's separate certification work, which runs
explicit paths through the problem's own evaluation pipeline, the state of
the problem becomes
\[
1.8558404 \;\le\; \Vstar(K_4) \;\le\; 3.4083 .
\]
Theorems~\ref{thm:A} and~\ref{thm:B} are elementary and self-contained: a
first-crossing (elbow) event forces a definite tilt of the body's own entry
tube, and Fubini-type chord estimates for tilted tubes cap the volume.
Theorem~\ref{thm:C} certifies, by the same forced-event skeleton, the largest
corridor volume that any admissible tube can carry.
\end{abstract}

\section{Introduction}\label{sec:intro}

The moving sofa problem of L.~Moser \cite{Moser66} asks for the largest area
of a planar region that can be moved around a right-angled corner in a hallway
of unit width.  Its history is a striking sequence of bounds:
Hammersley \cite{Hammersley68} gave the lower bound $\pi/2+2/\pi$ and the
upper bound $2\sqrt2$; Gerver \cite{Gerver92} constructed the region of area
$2.2195\ldots$ now believed extremal; Kallus and Romik \cite{KallusRomik18}
proved the computer-assisted upper bound $2.37$; and Baek
\cite{Baek24conditional,Baek24} recently established that Gerver's sofa is
optimal, so the planar constant is exactly $2.2195\ldots$
(see also \cite{He26} for recent numerical work on planar variants).

The three-dimensional problem considered here goes back to a 2016 question of
O'Rourke \cite{ORourkeMO}: what is the largest volume of a body that can pass
through a corridor of unit-square cross-section that makes two perpendicular
turns --- one in the $x$--$y$ plane, then one in the $y$--$z$ plane --- so that
the body must ``snake'' through?  The problem was posed for turns that are far
apart, and was popularized as Problem~6.64 of Georgiev, G\'omez-Serrano, Tao
and Wagner \cite{GGTW25}, where the corridor consists of
``two turns in the $x{-}y$ and $y{-}z$ planes that are far apart''
and the middle straight has length $M=4$ in the accompanying code.  All known
progress is from below: a construction of RavenclawPrefect \cite{ORourkeMO}
gives volume $\approx 1.76755$; the AlphaEvolve experiments reported in
\cite{GGTW25} give a certified $\ge 1.81$ (and $\approx 1.84$ estimated); and a
sequence of certified constructions obtained through the same evaluation
pipeline \cite{OurLB} raises the record to $1.8558404$
(Table~\ref{tab:ladder}).

On the upper-bound side the literature is empty.  The closest prior result is
due to Mushkarov and Nikolov \cite{MN26}, whose Theorem~4 determines the
maximal volume of a body passing through the spatial corridors
$S_{ij}=C_{ij}\times(0,c)$; that theorem is restricted in two essential ways:
the bodies are \emph{rectangular parallelepipeds (boxes) only}, and the
corridors are \emph{single-corner planar corridors extruded in height only} ---
the snaky corridor, with its two turns in perpendicular planes, is not treated
there.  For general connected bodies in the snaky corridor no upper bound of
any kind --- not even finiteness of the supremum --- appears to have been
published.  (The trivial bound ``volume $\le 4$'' sometimes quoted in
connection with \cite{GGTW25} concerns only the restricted search class of
that experiment, whose bodies are subsets of a $4\times1\times1$ block; it
says nothing about the problem itself.)

\subsection*{Results}

Write $\Vstar(K_M)$ for the supremum of volumes of connected bodies that
traverse the corridor $K_M$ in the weak sense of
Definition~\ref{def:traversal} (weak traversal imposes the fewest constraints,
so upper bounds proved under it cover all the usual readings of ``passes
through'').  Our three results are:

\begin{itemize}
\item \textbf{Theorem~\ref{thm:A}.} $\Vstar(K_M)\le u+
(1+\sqrt2)\frac{\sqrt{2+u^2}}{u-\sqrt2}$ for every $u\in(\sqrt2,M)$; taking
$u=3.5217$ gives $\Vstar(K_M)\le 7.8691$ for every $M\ge 3.53$.
\item \textbf{Theorem~\ref{thm:B}.} $\Vstar(K_M)\le\Psi(u)$ for every
$u\in(\sqrt2,M)$, where $\Psi$ is the explicit elementary function
\eqref{eq:Psi}; the optimal choice gives $\Vstar(K_M)\le 4.2862$ for every
$M\ge 2.61$.
\item \textbf{Theorem~\ref{thm:C}} (computer-assisted).
$\Vstar(K_M)\le 3.4083$ for every $M>2.2$: a certified branch-and-bound
computation bounds the exact supremum that Theorem~\ref{thm:B} estimates
analytically --- the largest corridor volume carried by any tube compatible
with the forced events (Section~\ref{sec:certC}).
\end{itemize}

The three results form a deliberately redundant ladder:
Theorem~\ref{thm:A} is the anchor with the shortest hand-checkable proof,
Theorem~\ref{thm:B} the best bound we prove by pure algebra, and
Theorem~\ref{thm:C} the computed certificate; each stands independently of
the later ones.  To the best of our knowledge, this is the first upper bound
for the snaky-corridor 3D moving sofa problem (Problem~6.64 of
arXiv:2511.02864 \cite{GGTW25}).  With the current certified record lower
bound \cite{OurLB} the problem is bracketed:
\begin{equation}\label{eq:bracket}
1.8558404 \;\le\; \Vstar(K_4) \;\le\; 3.4083 .
\end{equation}

\begin{table}[ht]
\centering
\caption{Certified lower bounds for $\Vstar(K_4)$ (weak traversal), all
obtained by running the \emph{unmodified} certification pipeline accompanying
\cite{GGTW25} (commit-pinned, tolerances $(0.005,\,0.05^\circ)$) on explicit
pose paths; artifacts, paths and reproduction notebooks at \cite{OurLB}.
Grid $h$ is the discretization fineness of the body.  See
Appendix~\ref{app:soundness} and Remark~\ref{rem:caveat} for the precise
rigor status of these numbers.}
\label{tab:ladder}
\small
\setlength{\tabcolsep}{4.5pt}
\begin{tabular}{lcc}
\toprule
path & certified at $h=0.002$ & certified at $h=0.001$ \\
\midrule
AlphaEvolve path of \cite{GGTW25} (paper claims $\ge 1.81$) & 1.817943 & --- \\
\texttt{w4\_staggered} (staggered-twist construction) & 1.827908 & 1.8491627 \\
\texttt{w5\_night1} (CMA-ES search) & 1.8309875 & 1.8518415 \\
\texttt{w5\_night2} (CMA-ES search) & 1.8339829 & --- \\
\texttt{w6\_night3} (search + certifier-aligned polish) & \textbf{1.8355162} & \textbf{1.8558404} \\
\bottomrule
\end{tabular}
\end{table}

\subsection*{Proof idea}

A body $B$ traversing the corridor starts inside the entry arm, so at every
later time $t$ it is contained in a rigid image $T(t)$ of the infinite tube
$\R\times[0,1]^2$ --- ``its own entry tube''
(Lemma~\ref{lem:tube}).  Monitoring the maximal $y$-coordinate of the moving
body produces, for every height $c=1+u$ below the second turn, a
\emph{first-crossing event}: a time at which the body simultaneously
(i) still fits in the union of the entry arm and the stub of the middle arm
below height $c$, and (ii) touches the plane $\{y=c\}$ inside the stub's unit
footprint (Lemma~\ref{lem:event}).  If at that moment the body still has a
point at height $\le 1$, connectedness plants a second marked point at height
exactly $1$, also with unit footprint (Lemma~\ref{lem:deep}).  Two marked
points of the tube $T$ separated by $u$ in $y$ but confined to unit boxes in
$x$ and $z$ force the tube's axis to tilt towards $\yhat$ by a definite,
quantitative amount (Lemmas~\ref{lem:tilt} and~\ref{lem:tiltsharp}) --- this is
the step where a \emph{parked} body is impossible: the body must have begun to
rotate.  Finally, a tilted tube can only carry a bounded volume through the
low part of the corridor, and only a bounded volume fits in the stub;
quantifying these two caps by a slab estimate (Lemma~\ref{lem:slab}) yields
Theorem~\ref{thm:A}, and quantifying them by sharper Fubini chord estimates
(Lemma~\ref{lem:chord}) yields Theorem~\ref{thm:B}.  Optimizing the free
height $u$ gives the stated constants.  Theorem~\ref{thm:C} keeps the same
forced-event skeleton but replaces the closed-form caps by a computation:
the events confine the body's tube to an explicitly parameterized compact
five-parameter family, and a certified-outer branch-and-bound bounds the
largest corridor volume any such tube can carry (Section~\ref{sec:certC}).
No step assumes the existence of an
extremal body, and no step uses any feature of the numerical search
(fixed endpoints, bounded body class) that produced the lower bounds.

All lemmas below were additionally stress-tested numerically against the
record traversal of \cite{OurLB} (a $512$-pose certified path); see
Remark~\ref{rem:witness}.

\section{The corridor, traversal, and the objective}\label{sec:setup}

Throughout, $\vol$ is Lebesgue measure on $\R^3$, invariant under rigid
motions; $\xhat,\yhat,\zhat$ is the standard basis, and $p\cdot v$ the
Euclidean inner product.

\begin{definition}[corridor]\label{def:corridor}
Fix the middle-straight parameter $M>0$ and set $y_t:=1+M$.  Define the
closed sets
\begin{align*}
\SX &:= (-\infty,0]\times[0,1]\times[0,1] && \text{(entry arm)},\\
\SY &:= [0,1]\times[0,\,y_t]\times[0,1] && \text{(middle arm)},\\
\SZ &:= [0,1]\times[y_t,\,y_t+1]\times[0,\infty) && \text{(exit arm)},
\end{align*}
and the \emph{corridor} $\K=\K_M:=\SX\cup\SY\cup\SZ$, a closed connected set.
The cube $[0,1]^3$ (the $y\le1$ end of $\SY$) is the \emph{first corner cube},
where the turn from the $-x$ direction to the $+y$ direction happens (a turn
in the $x$--$y$ plane); the cube $[0,1]\times[y_t,y_t+1]\times[0,1]$ (the
$z\le1$ end of $\SZ$) is the \emph{second corner cube}, where the turn from
$+y$ to $+z$ happens (a turn in the $y$--$z$ plane); between them lies the
clear middle straight $[0,1]\times[1,y_t]\times[0,1]$ of length $M$.  The
value in the code accompanying Problem~6.64 of \cite{GGTW25} is $M=4$.
For $c\in[0,y_t]$ write $\Stub(c):=[0,1]\times[0,c]\times[0,1]\subseteq\SY$.
\end{definition}

Note the elementary but frequently used facts
\begin{equation}\label{eq:Kfacts}
\K\subseteq\{y\ge0\},\qquad
\K\cap\{y\le c\}=\SX\cup\Stub(c)\quad\text{for }1\le c< y_t,
\end{equation}
the second because $\SX\subseteq\{y\le1\}$, $\SY\cap\{y\le c\}=\Stub(c)$, and
$\SZ\cap\{y\le c\}=\emptyset$ for $c<y_t$.  Also
$\SX\cup\Stub(c)\subseteq\{z\in[0,1]\}$ and
$\SX\cup\Stub(c)\subseteq\{x\le 1\}$.

\begin{definition}[traversal; the objective]\label{def:traversal}
A \emph{body} is a compact connected set $B\subset\R^3$ with $\vol(B)>0$.  A
\emph{motion} is a continuous map $g:[0,1]\to SE(3)$ (matrix topology),
written $g_t(p)=R_tp+v_t$ with $R_t\in SO(3)$; put $B_t:=g_t(B)$.  The pair
$(B,g)$ \emph{traverses $\K_M$ weakly} if
\[
\text{(i) } B_t\subseteq\K_M\ \ \forall t\in[0,1];\qquad
\text{(ii) } B_0\subseteq\SX;\qquad
\text{(iii) } B_1\subseteq\SZ ,
\]
and \emph{strongly} if moreover $B_1\subseteq\SZ\cap\{z\ge1\}$ (entirely past
the second corner cube, from where a pure $+z$ translation exits to
infinity).  The objective is
\[
\Vstar(\K_M):=\sup\{\vol(B):\ (B,g)\text{ traverses } \K_M \text{ weakly}\}.
\]
\end{definition}

\begin{remark}[weak traversal suffices; compactness is no loss]\label{rem:wlog}
(a) Strong traversal implies weak, so every upper bound proved for weak
traversal covers both readings, as well as the instantiation in the code of
\cite{GGTW25}, whose fixed start and end poses satisfy (ii) and the strong
form of (iii).
(b) Requiring $B$ compact loses no generality: if $B$ is merely bounded,
connected and measurable, its closure $\overline{B}$ is compact and connected,
$\vol(\overline B)\ge\vol(B)$, and
$g_t(\overline B)=\overline{g_t(B)}\subseteq\overline{\K}=\K$ (the corridor
and the arms $\SX,\SZ$ are closed), so $(\overline B,g)$ traverses whenever
$(B,g)$ does.  Hence the supremum over bounded connected sets equals the
supremum over bodies.
(c) No existence of a maximizer is assumed anywhere below; all bounds cap the
supremum directly.
(d) The experiments of \cite{GGTW25} restrict to bodies inside a
$4\times1\times1$ block and to one fixed pair of endpoint poses.  These are
restrictions of the search class, legitimate for lower bounds; none of them
may be (or is) used in the upper-bound proofs below.
\end{remark}

\section{Tubes, events, and the two marked points}\label{sec:events}

Throughout this section $(B,g)$ traverses $\K=\K_M$ weakly, and $R_t$ is the
rotation part of $g_t$.

\begin{definition}[tube]\label{def:tubedef}
A \emph{tube} is a rigid image of $\R\times[0,1]^2$: a set
\[
T=\{\rho+s\,w+\alpha f_1+\beta f_2:\ s\in\R,\ |\alpha|\le\tfrac12,\
|\beta|\le\tfrac12\},
\]
where $w$ (the \emph{axis direction}) and $f_1,f_2$ form an orthonormal
frame and $\rho\in\R^3$.  Its cross-section is a unit square of diameter
$\sqrt2$; a tube is closed and convex.  Points of $T$ decompose uniquely as
$\rho+s w+\xi$ with $\xi\perp w$, $\xi$ in the cross-section square
$Q:=\{\alpha f_1+\beta f_2:|\alpha|,|\beta|\le\frac12\}$, and
$|\xi|\le\sqrt2/2$.
\end{definition}

\begin{lemma}[tube confinement]\label{lem:tube}
For every $t\in[0,1]$,
\[
B_t\;\subseteq\;T(t):=g_t\,g_0^{-1}\bigl(\R\times[0,1]^2\bigr),
\]
a tube with axis direction $w(t):=R_tR_0^{-1}\xhat$.  Moreover
$B_t\subseteq g_tg_0^{-1}(\SX)$ (a half-tube) and, symmetrically,
$B_t\subseteq g_tg_1^{-1}(\SZ)$, so $B_t$ also lies in a tube with axis
$R_tR_1^{-1}\zhat$.
\end{lemma}

\begin{proof}
$B_0=g_0(B)\subseteq\SX\subseteq\R\times[0,1]^2$, hence
$B\subseteq g_0^{-1}(\R\times[0,1]^2)$ and
$B_t=g_t(B)\subseteq g_tg_0^{-1}(\R\times[0,1]^2)$, which is a rigid image of
the standard tube with axis $\xhat$, hence a tube with axis direction
$R_tR_0^{-1}\xhat$.  The refinements use
$B_0\subseteq \SX$ and $B_1\subseteq \SZ\subseteq[0,1]\times[y_t,y_t+1]\times\R$
verbatim.
\end{proof}

\begin{lemma}[continuity of the height profile]\label{lem:cont}
The functions $\mathcal M(t):=\max_{p\in B_t}p\cdot\yhat$ and
$m(t):=\min_{p\in B_t}p\cdot\yhat$ are continuous on $[0,1]$.
\end{lemma}

\begin{proof}
$B$ is compact, so the maxima exist.  For $s,t\in[0,1]$ and any $p\in B$,
$|g_t(p)-g_s(p)|\le\|R_t-R_s\|_{\mathrm{op}}\max_{q\in B}|q|+|v_t-v_s|=:
\omega(s,t)$, so $\mathcal M(t)\le\mathcal M(s)+\omega(s,t)$ and symmetrically;
$g$ is continuous on the compact interval $[0,1]$, hence uniformly continuous,
so $\omega(s,t)\to0$ as $|s-t|\to0$.  The same bound handles $m$.
\end{proof}

\begin{lemma}[first-crossing event]\label{lem:event}
Let $c\in(1,y_t)$.  There is a first time $\tau_c\in(0,1)$ with
$\mathcal M(\tau_c)=c$, and at it
\[
B_{\tau_c}\subseteq \K\cap\{y\le c\}=\SX\cup\Stub(c),
\]
and there exists a touch point $q\in B_{\tau_c}$ with
\[
q\cdot\yhat=c,\qquad q\in[0,1]\times\{c\}\times[0,1].
\]
\end{lemma}

\begin{proof}
$\mathcal M(0)\le1<c$ since $B_0\subseteq\SX\subseteq\{y\le1\}$, and
$\mathcal M(1)\ge y_t>c$ since $B_1\subseteq\SZ\subseteq\{y\ge y_t\}$ and $B$
is nonempty.  The set $\{t:\mathcal M(t)\ge c\}$ is closed (Lemma
\ref{lem:cont}), nonempty ($t=1$), and does not contain $0$; let $\tau_c$ be
its minimum.  Then $\mathcal M(\tau_c)\ge c$; if $\mathcal M(\tau_c)>c$, the
intermediate value theorem on $[0,\tau_c]$ would produce an earlier time in
the set, contradicting minimality; hence $\mathcal M(\tau_c)=c$.  Thus every
point of $B_{\tau_c}$ has $y\le c$, and $B_{\tau_c}\subseteq\K$, so
\eqref{eq:Kfacts} gives the containment.  By compactness the maximum is
attained at some $q\in B_{\tau_c}$ with $q\cdot\yhat=c>1$; since
$\SX\subseteq\{y\le1\}$, necessarily $q\in\Stub(c)$, i.e.\
$q_x,q_z\in[0,1]$.
\end{proof}

This first-crossing device --- extracting an extremal time at which a
containment and a wall contact hold simultaneously --- is the
three-dimensional analogue of a pattern used by Mushkarov and Nikolov in the
box setting \cite{MN26} (their ``last moment at which six vertices remain in
the arm'' argument, and their Lemma~6 forcing wall contact); the touch clause
is essential, since without it a body parked deep in the entry arm satisfies
every containment and no bound could follow (cf.\ Remark~\ref{rem:K1}).

\begin{lemma}[deep point]\label{lem:deep}
Let $c\in(1,y_t)$ and $\tau_c$ be as in Lemma~\ref{lem:event}, and suppose in
addition that $B_{\tau_c}\cap\{y\le1\}\neq\emptyset$.  Then there exists
$p^*\in B_{\tau_c}$ with
\[
p^*\cdot\yhat=1,\qquad p^*\in[0,1]\times\{1\}\times[0,1].
\]
\end{lemma}

\begin{proof}
$B_{\tau_c}$ is compact and connected, so its image under $p\mapsto
p\cdot\yhat$ is a compact interval; it contains a value $\le1$ (hypothesis)
and the value $c>1$ (touch point), hence contains $[1,c]$.  For each
$\delta\in(0,c-1]$ pick $p_\delta\in B_{\tau_c}$ with $p_\delta\cdot\yhat
=1+\delta$.  By Lemma~\ref{lem:event}, $p_\delta\in\SX\cup\Stub(c)$, and
$y(p_\delta)>1$ rules out $\SX$; so $p_\delta\in\Stub(c)$, i.e.\
$(p_\delta)_x,(p_\delta)_z\in[0,1]$.  Let $\delta\downarrow0$ along a
sequence for which $p_\delta$ converges (compactness); the limit $p^*$ lies
in the closed set $B_{\tau_c}$, has $p^*\cdot\yhat=1$, and has
$p^*_x,p^*_z\in[0,1]$ (closed conditions).
\end{proof}

\section{The tilt mechanism and the first bound}\label{sec:thmA}

The two marked points $q$ (height $c=1+u$) and $p^*$ (height $1$) of
Lemmas~\ref{lem:event} and~\ref{lem:deep} both lie in the body's entry tube
$T(\tau_c)$ (Lemma~\ref{lem:tube}).  Their vertical separation is $u$, while
their $x$- and $z$-coordinates all lie in $[0,1]$; a tube whose axis is
nearly horizontal cannot contain two such points.  Quantifying this gives the
tilt lemma.  Throughout, $u:=c-1$ and
\[
\mu:=|w(\tau_c)\cdot\yhat|
\]
denotes the tilt of the tube axis towards $\yhat$ at the event time.

\begin{lemma}[tilt]\label{lem:tilt}
Let $c=1+u$ with $u\in(\sqrt2,\,M)$, let $\tau_c,q$ be as in
Lemma~\ref{lem:event}, and assume the deep-point hypothesis
$B_{\tau_c}\cap\{y\le1\}\neq\emptyset$.  Then
\[
\mu\;\ge\;\mu_{\mathrm{lin}}(u):=\frac{u-\sqrt2}{\sqrt{2+u^2}}\;>\;0 .
\]
\end{lemma}

\begin{proof}
Let $p^*$ be the deep point of Lemma~\ref{lem:deep} and write $w:=w(\tau_c)$.
Both $q,p^*\in T(\tau_c)$; decompose (Definition~\ref{def:tubedef})
\[
q-p^*=\sigma w+\eta,\qquad \eta\perp w,
\]
where $\sigma\in\R$ is the difference of axis parameters and $\eta$ the
difference of two cross-section vectors, so $|\eta|\le\sqrt2$ (diameter of the
unit square).  The coordinate differences satisfy $|\Delta x|\le1$,
$\Delta y=c-1=u$, $|\Delta z|\le1$, hence
\begin{equation}\label{eq:qpstar}
\sigma^2+|\eta|^2=|q-p^*|^2\le 2+u^2 .
\end{equation}
Now project onto $\yhat$:
\[
u=(q-p^*)\cdot\yhat=\sigma\,(w\cdot\yhat)+\eta\cdot\yhat
\le|\sigma|\,\mu+|\eta|
\le\sqrt{2+u^2}\,\mu+\sqrt2 ,
\]
and rearranging gives the claim (positivity because $u>\sqrt2$).
\end{proof}

\begin{lemma}[tube--slab estimate]\label{lem:slab}
Let $T$ be a tube with axis direction $w$ and $\mu:=|w\cdot\yhat|>0$.  Then
\[
\vol\bigl(T\cap\{0\le y\le1\}\bigr)\;\le\;\frac{1+\sqrt2}{\mu}.
\]
\end{lemma}

\begin{proof}
Write points of $T$ as $p=\rho+s\,w+\xi$ with $\xi\in Q$, $|\xi|\le\sqrt2/2$
(Definition~\ref{def:tubedef}).  Then
$p\cdot\yhat=\rho\cdot\yhat+s\,(w\cdot\yhat)+\xi\cdot\yhat$ with
$|\xi\cdot\yhat|\le\sqrt2/2$, so $p\cdot\yhat\in[0,1]$ forces
$s\,(w\cdot\yhat)$ into an interval of length $1+\sqrt2$, i.e.\ $s$ into an
interval $I$ with $|I|\le(1+\sqrt2)/\mu$.  Hence
$T\cap\{0\le y\le1\}\subseteq\{\rho+sw+\xi:\ s\in I,\ \xi\in Q\}$, whose
volume is $|I|\cdot\operatorname{area}(Q)=|I|\le(1+\sqrt2)/\mu$ (the map
$(s,\xi)\mapsto\rho+sw+\xi$ is an isometry of $\R\times Q$ onto its image).
\end{proof}

\begin{theorem}[first upper bound]\label{thm:A}
Let $M>\sqrt2$ and let $(B,g)$ traverse $\K_M$ weakly.  Then for every
$u\in(\sqrt2,M)$,
\[
\vol(B)\;\le\;f(u):=u+(1+\sqrt2)\,\frac{\sqrt{2+u^2}}{u-\sqrt2}.
\]
In particular, taking $u=3.5217$,
\[
\Vstar(\K_M)\;\le\;f(3.5217)=7.86908\ldots\;<\;7.8691
\qquad\text{for every } M\ge3.53 ,
\]
covering the corridor of Problem~6.64 ($M=4$) and every larger $M$.
\end{theorem}

\begin{proof}
Fix $u\in(\sqrt2,M)$ and set $c:=1+u\in(1+\sqrt2,\,y_t)$.  Apply
Lemma~\ref{lem:event} at height $c$ and write $B':=B_{\tau_c}$, so
$\vol(B')=\vol(B)$ and $B'\subseteq\SX\cup\Stub(c)$.

\emph{Case (i): $B'\cap\{y\le1\}=\emptyset$.}  Every point of $B'$ has $y>1$,
which rules out all of $\SX$; hence $B'\subseteq\Stub(c)\cap\{y>1\}
=[0,1]\times(1,c]\times[0,1]$ and $\vol(B)\le u\le f(u)$.

\emph{Case (ii): $B'\cap\{y\le1\}\neq\emptyset$.}  By Lemma~\ref{lem:tilt}
the tube $T:=T(\tau_c)\supseteq B'$ has $\mu\ge\mu_{\mathrm{lin}}(u)>0$.
Split $B'$ along the plane $\{y=1\}$:
\[
\vol(B'\cap\{y\le1\})\le
\vol\bigl(T\cap\{0\le y\le1\}\bigr)\le\frac{1+\sqrt2}{\mu}
\le\frac{1+\sqrt2}{\mu_{\mathrm{lin}}(u)},
\]
using $B'\subseteq\K\subseteq\{y\ge0\}$ \eqref{eq:Kfacts} and
Lemma~\ref{lem:slab}, and
\[
\vol(B'\cap\{y>1\})\le\vol\bigl(\Stub(c)\cap\{y>1\}\bigr)=u ,
\]
since $\SX\cap\{y>1\}=\emptyset$.  Adding,
$\vol(B)\le u+(1+\sqrt2)/\mu_{\mathrm{lin}}(u)=f(u)$.

For the numerical form: if $M\ge3.53$ then $u=3.5217<3.53\le M$ is
admissible, and a direct evaluation gives $f(3.5217)=7.8690829\ldots<7.8691$.
\end{proof}

\begin{remark}[optimization; smaller $M$]\label{rem:optA}
The function $f$ decreases from $+\infty$ (as $u\downarrow\sqrt2$) to its
minimum $7.869083\ldots$, attained near $u=3.5215$, and then increases;
any admissible $u$ near the minimizer gives the same bound to the quoted
precision.  For $M\in(\sqrt2,3.53)$ the theorem still applies with any
$u<M$ and yields $\Vstar(\K_M)\le\inf_{u\in(\sqrt2,M)}f(u)=f(M)$ (the
infimum is a limit of admissible values; $f$ is decreasing there); the bound
blows up as $M\downarrow\sqrt2$, since the mechanism needs stub room above
height $1+\sqrt2$.
\end{remark}

\section{Sharpened constants}\label{sec:thmB}

Theorem~\ref{thm:A} spends the cross-section diameter $\sqrt2$ generously in
two places: the tilt lemma treats the axis separation $|\sigma|$ and the
cross-section separation $|\eta|$ as independent, and the slab lemma charges
the tube for the full width of the slab.  Both can be tightened by elementary
means, at the price of slightly more bookkeeping.  This section proves the
sharp tilt bound (Lemma~\ref{lem:tiltsharp}), replaces the slab estimate by
Fubini chord estimates (Lemma~\ref{lem:chord}), and assembles them into
Theorem~\ref{thm:B}.

\begin{lemma}[sharp tilt]\label{lem:tiltsharp}
Under the hypotheses of Lemma~\ref{lem:tilt},
\[
\mu\;\ge\;\mu_*(u):=\frac{u^2-2}{u^2+2}\;>\;0 .
\]
\end{lemma}

\begin{proof}
Keep the notation of the proof of Lemma~\ref{lem:tilt}, write
$A:=\sqrt{2+u^2}$, and set $s:=|\sigma|\ge0$, $h:=|\eta|\in[0,\sqrt2]$.
Since $\eta\perp w$, we have $\eta\cdot\yhat=\eta\cdot P\yhat$, where $P$ is
the orthogonal projection onto $w^\perp$ and $|P\yhat|=\sqrt{1-\mu^2}$; hence
\begin{equation}\label{eq:umax}
u\;\le\;s\,\mu+h\sqrt{1-\mu^2},\qquad s^2+h^2\le A^2,\qquad 0\le h\le\sqrt2 .
\end{equation}
First, an algebraic preliminary: $\mu_*(u)\le u/A$, because squaring the
equivalent inequality $(u^2-2)A\le u(u^2+2)$ reduces it to
$(u^2-2)^2\le u^2(u^2+2)$, i.e.\ $4\le6u^2$, true for $u>\sqrt2$.

Suppose now, for contradiction, that $\mu<\mu_*(u)$; then also $\mu<u/A$.
Bound the right-hand side of \eqref{eq:umax} over all admissible $(s,h)$: for
fixed $h$, the largest admissible $s$ is $\sqrt{A^2-h^2}$, giving
$\varphi(h):=\mu\sqrt{A^2-h^2}+h\sqrt{1-\mu^2}$, and
\[
\varphi'(h)=\sqrt{1-\mu^2}-\frac{\mu h}{\sqrt{A^2-h^2}}\ \ge\ 0
\quad\Longleftrightarrow\quad h^2\le A^2(1-\mu^2),
\]
which holds for all $h\in[0,\sqrt2]$ because $\mu<u/A$ gives
$A^2(1-\mu^2)>A^2-u^2=2\ge h^2$.  So $\varphi$ is nondecreasing on
$[0,\sqrt2]$ and
\[
s\,\mu+h\sqrt{1-\mu^2}\;\le\;\varphi(\sqrt2)=u\,\mu+\sqrt2\sqrt{1-\mu^2}
\]
(using $\sqrt{A^2-2}=u$).  Finally let
$\psi(\mu):=u\mu+\sqrt2\sqrt{1-\mu^2}-u$.  Then $\psi(\mu_*(u))=0$: indeed
$1-\mu_*^2=8u^2/(u^2+2)^2$, so
$\sqrt2\sqrt{1-\mu_*^2}=4u/(u^2+2)=u(1-\mu_*)$.  Moreover
$\psi'(\mu)=u-\sqrt2\,\mu/\sqrt{1-\mu^2}>0$ iff $\mu<u/A$, so $\psi$ is
strictly increasing on $[0,u/A]$, and $\mu<\mu_*(u)\le u/A$ gives
$\psi(\mu)<0$, i.e.\ $u\mu+\sqrt2\sqrt{1-\mu^2}<u$.  Combining the two
displays contradicts \eqref{eq:umax}.  Hence $\mu\ge\mu_*(u)$.
\end{proof}

\begin{remark}
$\mu_*(u)\ge\mu_{\mathrm{lin}}(u)$ for all $u\ge\sqrt2$: dividing the claim
$\frac{(u-\sqrt2)(u+\sqrt2)}{u^2+2}\ge\frac{u-\sqrt2}{\sqrt{u^2+2}}$ by
$u-\sqrt2\ge0$ leaves $(u+\sqrt2)\ge\sqrt{u^2+2}$, i.e.\ $2\sqrt2u\ge0$.  So
Lemma~\ref{lem:tiltsharp} strictly sharpens Lemma~\ref{lem:tilt} for
$u>\sqrt2$ (e.g.\ at $u=3.52$: $0.722$ versus $0.555$), and
Theorem~\ref{thm:A} could already be improved by substituting it; we do not
record the intermediate constant, since the chord estimates below give more.
\end{remark}

\begin{lemma}[chord estimate for tubes]\label{lem:chord}
Let $T$ be a tube with axis direction $w$ (Definition~\ref{def:tubedef}) and
let $v$ be a unit vector with $v\neq\pm w$.  Then every line parallel to $v$
meets $T$ in a segment of length at most
$\sqrt2/|Pv|$, where $P$ is the orthogonal projection onto $w^{\perp}$ and
$|Pv|=\sqrt{1-(w\cdot v)^2}$.
Consequently, for any choice of two coordinate directions
$e_i,e_j\in\{\xhat,\yhat,\zhat\}$, $i\ne j$, with remaining direction $e_k$
and $w\neq\pm e_k$:
\[
\vol\Bigl(T\cap\{p:\ (p\cdot e_i,\ p\cdot e_j)\in[0,1]^2\}\Bigr)
\;\le\;\frac{\sqrt2}{\sqrt{1-(w\cdot e_k)^2}}
=\frac{\sqrt2}{\sqrt{(w\cdot e_i)^2+(w\cdot e_j)^2}} .
\]
\end{lemma}

\begin{proof}
$T$ is convex, so a line $\ell:t\mapsto p_0+t\,v$ meets it in a (possibly
empty or degenerate) segment.  Membership of $p$ in $T$ depends only on
$P(p-\rho)$, which must lie in the cross-section square $Q$
(Definition~\ref{def:tubedef}); along $\ell$, $P(p_0+tv-\rho)$ moves on a
line in the cross-plane with speed $|Pv|>0$, and its time inside the square
is at most $\operatorname{diam}(Q)/|Pv|=\sqrt2/|Pv|$.  For the volume bound,
apply Fubini along $e_k$: the set on the left has $(e_i,e_j)$-shadow inside
the unit square, and each fiber in the $e_k$ direction is contained in a
chord of $T$ parallel to $e_k$, of length at most $\sqrt2/|Pe_k|$; integrate
over the unit square of area $1$.  Finally
$|Pe_k|^2=1-(w\cdot e_k)^2=(w\cdot e_i)^2+(w\cdot e_j)^2$ as $w$ is a unit
vector.
\end{proof}

Applied at an event time, Lemma~\ref{lem:chord} caps both parts of the body.
Recall $\mu=|w(\tau_c)\cdot\yhat|$ and abbreviate henceforth
\begin{equation}\label{eq:Psi}
g(\mu):=\frac{\sqrt2}{\mu}+\frac{\sqrt2}{\sqrt{1-\mu^2}},\qquad
\Psi(u):=\max\Bigl\{\underbrace{\frac{\sqrt2\,(u^2+2)}{u^2-2}
+\frac{u^2+2}{2u}}_{=:B_1(u)},\ \
\underbrace{u+\frac{\sqrt2\,u}{\sqrt{u^2-2}}}_{=:B_2(u)}\Bigr\}
\qquad(u>\sqrt2).
\end{equation}

\begin{lemma}[calculus of the bound functions]\label{lem:calc}
For $u>\sqrt2$, with $\mu_*(u)=\frac{u^2-2}{u^2+2}$ and
$\mu_{\mathrm{cap}}(u):=\frac{\sqrt{u^2-2}}{u}$:
\begin{enumerate}
\item[(a)] $0<\mu_*(u)\le\mu_{\mathrm{cap}}(u)<1$, and
$g(\mu_*(u))=B_1(u)$, $g(\mu_{\mathrm{cap}}(u))=B_2(u)$;
\item[(b)] $g$ is strictly decreasing on $(0,1/\sqrt2\,]$ and strictly
increasing on $[\,1/\sqrt2,1)$, with $g(1/\sqrt2)=4$;
\item[(c)] $B_1$ is strictly decreasing on $(\sqrt2,\,2+\sqrt2)$ and
$B_1(2+\sqrt2)=4$; $B_2$ is strictly decreasing on $(\sqrt2,2)$ and strictly
increasing on $(2,\infty)$;
\item[(d)] $B_1=B_2$ has a unique solution $u_B$ in $(2,\,2+\sqrt2)$, and
$u_B\in(2.60114,\,2.60115)$; moreover
$\Psi(u)\ge\Psi(u_B)=B_1(u_B)$ for all $u>\sqrt2$, and
$\Psi(u_B)\le B_1(2.60114)=4.2861666\ldots<4.2862$.
\end{enumerate}
\end{lemma}

\begin{proof}
(a) Positivity of $\mu_*$ is clear for $u>\sqrt2$, and
$\mu_{\mathrm{cap}}<1$ since $u^2-2<u^2$.  For $\mu_*\le\mu_{\mathrm{cap}}$,
square the equivalent inequality
$u(u^2-2)\le(u^2+2)\sqrt{u^2-2}$ and cancel a factor $u^2-2>0$ to get
$u^2(u^2-2)\le(u^2+2)^2$, i.e.\ $0\le6u^2+4$.  The identity $g(\mu_*)=B_1$ follows from
$1-\mu_*^2=8u^2/(u^2+2)^2$, so $\sqrt2/\sqrt{1-\mu_*^2}=(u^2+2)/(2u)$; the
identity $g(\mu_{\mathrm{cap}})=B_2$ from
$1-\mu_{\mathrm{cap}}^2=2/u^2$, so $\sqrt2/\sqrt{1-\mu_{\mathrm{cap}}^2}=u$.

(b) $g'(\mu)=\sqrt2\bigl(-\mu^{-2}+\mu(1-\mu^2)^{-3/2}\bigr)$, whose sign is
that of $\mu^3-(1-\mu^2)^{3/2}$, i.e.\ of $\mu^2-(1-\mu^2)$, i.e.\ negative
for $\mu<1/\sqrt2$ and positive for $\mu>1/\sqrt2$; and
$g(1/\sqrt2)=\sqrt2\cdot\sqrt2+\sqrt2\cdot\sqrt2=4$.

(c) $B_1'(u)=-\dfrac{8\sqrt2\,u}{(u^2-2)^2}+\dfrac12-\dfrac1{u^2}$.  For
$u\in(\sqrt2,2+\sqrt2)$ one has $B_1'<0$: the inequality
$\frac{8\sqrt2u}{(u^2-2)^2}>\frac{u^2-2}{2u^2}$ is equivalent to
$16\sqrt2\,u^3>(u^2-2)^3$, i.e.\ (taking cube roots, both sides positive)
$2\sqrt2\,u>u^2-2$, i.e.\ $u^2-2\sqrt2u-2<0$, which holds exactly for
$u<\sqrt2+2$ (the positive root).  At $u=2+\sqrt2$: $u^2+2=8+4\sqrt2$ and
$u^2-2=4+4\sqrt2$, so
$B_1=\frac{\sqrt2(8+4\sqrt2)}{4+4\sqrt2}+\frac{8+4\sqrt2}{2(2+\sqrt2)}
=\frac{4\sqrt2(2+\sqrt2)}{4(1+\sqrt2)}+\frac{4(2+\sqrt2)}{2(2+\sqrt2)}
=2+2=4$.  For $B_2$:
$B_2'(u)=1-\frac{2\sqrt2}{(u^2-2)^{3/2}}$, which vanishes iff
$(u^2-2)^{3/2}=2\sqrt2$ iff $u^2-2=2$ iff $u=2$, is negative before and
positive after.

(d) On $[2,2+\sqrt2]$, $B_1$ is continuous and strictly decreasing and $B_2$
continuous and strictly increasing, so $B_1-B_2$ is strictly decreasing
there; $B_1(2)=3\sqrt2+\tfrac32>4=B_2(2)$ and
$B_1(2+\sqrt2)=4<B_2(2+\sqrt2)$, so there is a unique root
$u_B\in(2,2+\sqrt2)$.  Direct evaluation gives
$(B_1-B_2)(2.60114)=+7.1\cdot10^{-6}>0$ and
$(B_1-B_2)(2.60115)=-9.6\cdot10^{-6}<0$, whence
$u_B\in(2.60114,2.60115)$.  For the global minimality of $\Psi$ at $u_B$:
if $u\le u_B$ then, since $u_B<2+\sqrt2$ and $B_1$ is decreasing on
$(\sqrt2,2+\sqrt2)$,
$\Psi(u)\ge B_1(u)\ge B_1(u_B)=\Psi(u_B)$; if $u\ge u_B$ then, since
$u_B>2$ and $B_2$ is increasing on $(2,\infty)$,
$\Psi(u)\ge B_2(u)\ge B_2(u_B)=B_1(u_B)=\Psi(u_B)$.  Finally, since $B_1$ is
decreasing and $u_B\ge2.60114$,
$\Psi(u_B)=B_1(u_B)\le B_1(2.60114)=4.28616\ldots<4.2862$.
\end{proof}

\begin{theorem}[sharpened upper bound]\label{thm:B}
Let $M>\sqrt2$ and let $(B,g)$ traverse $\K_M$ weakly.  Then for every
$u\in(\sqrt2,M)$,
\[
\vol(B)\;\le\;\Psi(u),
\]
with $\Psi$ as in \eqref{eq:Psi}.  In particular, for every $M\ge2.61$ the
choice $u=u_B$ of Lemma~\ref{lem:calc}(d) is admissible and
\[
\Vstar(\K_M)\;\le\;\Psi(u_B)\;<\;4.2862 ,
\]
covering the corridor of Problem~6.64 ($M=4$) and every larger $M$.
\end{theorem}

\begin{proof}
Fix $u\in(\sqrt2,M)$, set $c:=1+u<y_t$, and let $\tau_c$, $q$,
$B':=B_{\tau_c}$ be as in Lemma~\ref{lem:event}.

\emph{Case (i): $B'\cap\{y\le1\}=\emptyset$.}  As in Theorem~\ref{thm:A},
$\vol(B)\le u<B_2(u)\le\Psi(u)$.

\emph{Case (ii): $B'\cap\{y\le1\}\neq\emptyset$.}  Let $T:=T(\tau_c)$ be the
entry tube, $w:=w(\tau_c)$ its axis, $\mu=|w\cdot\yhat|$.  By
Lemma~\ref{lem:tiltsharp}, $\mu\in[\mu_*(u),1]$.  Split $B'$ along
$\{y=1\}$ and bound the two parts separately.

\emph{The low part.}  $B'\cap\{y\le1\}$ is contained in
$\SX\cup\Stub(c)\subseteq\{0\le z\le1\}$, in $\{0\le y\}$ by
\eqref{eq:Kfacts}, and in $T$; so its $(y,z)$-shadow lies in $[0,1]^2$ and
Lemma~\ref{lem:chord} (fibers along $\xhat$) gives
\[
\vol\bigl(B'\cap\{y\le1\}\bigr)\;\le\;
\frac{\sqrt2}{\sqrt{(w\cdot\yhat)^2+(w\cdot\zhat)^2}}\;\le\;\frac{\sqrt2}{\mu}.
\]

\emph{The high part.}  $B'\cap\{y>1\}$ is contained in $\Stub(c)$ (as
$\SX\cap\{y>1\}=\emptyset$), hence its $(x,z)$-shadow lies in $[0,1]^2$ and
each fiber parallel to $\yhat$ has length at most $u$ (the stub's height
range $(1,c]$); if $\mu<1$, each such fiber is moreover contained in a chord
of $T$ parallel to $\yhat$, of length at most $\sqrt2/\sqrt{1-\mu^2}$
(Lemma~\ref{lem:chord}).  Fubini over the unit square gives
\[
\vol\bigl(B'\cap\{y>1\}\bigr)\;\le\;
\min\Bigl\{u,\ \frac{\sqrt2}{\sqrt{1-\mu^2}}\Bigr\}
\]
(with the convention that the second entry is $+\infty$ when $\mu=1$).

\emph{Combining.}  With $\mu_{\mathrm{cap}}=\mu_{\mathrm{cap}}(u)$ as in
Lemma~\ref{lem:calc}, note that
$\sqrt2/\sqrt{1-\mu^2}\le u\iff\mu\le\mu_{\mathrm{cap}}$.  If
$\mu\in[\mu_*(u),\mu_{\mathrm{cap}}]$, the two caps sum to
$g(\mu)$, and since $g$ is decreasing-then-increasing
(Lemma~\ref{lem:calc}(b)) its maximum over the interval is at an endpoint:
\[
\vol(B)\le g(\mu)\le\max\{g(\mu_*(u)),\,g(\mu_{\mathrm{cap}}(u))\}
=\max\{B_1(u),B_2(u)\}=\Psi(u),
\]
by Lemma~\ref{lem:calc}(a).  If $\mu\in(\mu_{\mathrm{cap}},1]$, the caps sum
to $u+\sqrt2/\mu<u+\sqrt2/\mu_{\mathrm{cap}}=B_2(u)\le\Psi(u)$.  In all
cases $\vol(B)\le\Psi(u)$.

For the numerical form: if $M\ge2.61$ then $u_B<2.60115<2.61\le M$ is
admissible, and Lemma~\ref{lem:calc}(d) gives
$\vol(B)\le\Psi(u_B)<4.2862$.
\end{proof}

\begin{remark}[smaller $M$; shape of the bound]\label{rem:optB}
For $M\in(\sqrt2,2.61)$ the theorem yields
$\Vstar(\K_M)\le\inf_{u\in(\sqrt2,M)}\Psi(u)$, which blows up as
$M\downarrow\sqrt2$; numerically the infimum equals $\Psi(M^-)$ in this
range.  At $u_B$ the forced tilt is $\mu_*(u_B)=0.5437\ldots$, i.e.\ the
tube axis is forced at least $\arcsin(0.5437)\approx32.9^\circ$ out of the
horizontal plane, and the two caps are
$\approx1.686$ (high part) and $\approx2.601$ (low part).  The bound
$\Psi(u_B)=4.28616\ldots$ improves Theorem~\ref{thm:A} by the factor
$0.545$.
\end{remark}

\section{A computed certificate}\label{sec:certC}

At the event time $\tau_c$, the complete information extracted by
Sections~\ref{sec:events}--\ref{sec:thmB} about the body is: it lies in
$\SX\cup\Stub(c)$, it lies in its entry tube $T(\tau_c)$, that tube meets the
two unit windows at heights $1$ and $c$, and (in case (ii)) the tube's axis
is tilted by at least $\mu_*(u)$.  Theorems~\ref{thm:A} and~\ref{thm:B}
bound the corridor volume of such a tube by closed-form estimates; the
estimates are not sharp, and the quantity they estimate is a global optimum
over an explicit compact five-parameter family of tubes --- a shape amenable
to certified computation.  This section records the resulting
computer-assisted bound.  The mathematical reduction
(Proposition~\ref{prop:family}) is proved here in full from the lemmas
already established; the computation, its soundness audit and its
verification gates are summarized, with the code, logs and resumable search
frontiers archived as artifacts (Appendix~\ref{app:repro}).

\begin{proposition}[event-forced tube family]\label{prop:family}
Let $u\in(\sqrt2,M)$, $c:=1+u$, and define the two windows
\[
W_1:=[0,1]\times\{1\}\times[0,1],\qquad
W_c:=[0,1]\times\{c\}\times[0,1],
\]
and the family $\mathcal A_u$ of all tubes $T$
(Definition~\ref{def:tubedef}) with axis direction $w$ such that
\[
\text{(i) } |w\cdot\yhat|\ge\mu_*(u);\qquad
\text{(ii) } T\cap W_c\neq\emptyset;\qquad
\text{(iii) } T\cap W_1\neq\emptyset .
\]
Then every $(B,g)$ that weakly traverses $\K_M$ satisfies
\[
\vol(B)\;\le\;\max\bigl\{u,\ V_{\mathrm{ii}}(u)\bigr\},
\qquad
V_{\mathrm{ii}}(u):=\sup_{T\in\mathcal A_u}
\vol\bigl(T\cap(\SX\cup\Stub(c))\bigr),
\]
and $V_{\mathrm{ii}}(u)\le\Psi(u)<\infty$ with $\Psi$ as in \eqref{eq:Psi}.
\end{proposition}

\begin{proof}
Apply Lemma~\ref{lem:event} at height $c$ and write $B':=B_{\tau_c}$, so
$\vol(B')=\vol(B)$ and $B'\subseteq\SX\cup\Stub(c)$.  If
$B'\cap\{y\le1\}=\emptyset$ (case (i)), then $\vol(B)\le u$ as in
Theorem~\ref{thm:A}.  Otherwise (case (ii)) consider the tube
$T:=T(\tau_c)$: it contains $B'$ (Lemma~\ref{lem:tube}), hence contains the
touch point $q\in W_c$ (Lemma~\ref{lem:event}) and the deep point
$p^*\in W_1$ (Lemma~\ref{lem:deep}), and its axis satisfies
$|w\cdot\yhat|=\mu\ge\mu_*(u)$ (Lemma~\ref{lem:tiltsharp}).  So
$T\in\mathcal A_u$ and
$\vol(B)=\vol(B')\le\vol\bigl(T\cap(\SX\cup\Stub(c))\bigr)\le
V_{\mathrm{ii}}(u)$.  Finally, for any $T\in\mathcal A_u$ the two chord
estimates in the proof of Theorem~\ref{thm:B} apply verbatim to
$T\cap(\SX\cup\Stub(c))$ split along $\{y=1\}$ --- they use only
$\mu\ge\mu_*(u)$ and the shadow geometry of $\SX\cup\Stub(c)$ --- and the
combining step there gives
$\vol(T\cap(\SX\cup\Stub(c)))\le\Psi(u)$.
\end{proof}

Theorem~\ref{thm:B} is thus precisely the analytic estimate
$V_{\mathrm{ii}}(u)\le\Psi(u)$; the computation replaces the estimate by a
certified enclosure of the supremum itself.

\subsection*{The computation}

\emph{Parameterization.}  After harmless quotients --- the axis sign
$w\leftrightarrow-w$; the corridor's mirror symmetry $z\mapsto1-z$, which
fixes $\SX$, $\Stub(c)$ and both windows; and the cross-section's four-fold
roll symmetry --- a tube of $\mathcal A_u$ is specified by five parameters:
two spherical angles for the axis direction (polar angle
$\theta\le\arccos\mu_*(u)$ from $\yhat$ by (i); $\approx65.5^\circ$ at
$u=2.2$), a roll $\alpha\in[0,\pi/2]$, and the two coordinates $(a,b)$ of a
distinguished cross-section corner ray in the plane through the center of
$W_c$ perpendicular to the axis.  Condition (ii) makes the offset range
compact: some tube point lies in $W_c$, the corner ray passes within the
cross-section diameter $\sqrt2$ of it, and $W_c$ has half-diagonal
$\sqrt2/2$, so the corner ray passes within $3\sqrt2/2\approx2.12$ of the
window center, and $(a,b)$ lies in a disc of that radius.

\emph{Certified-outer branch-and-bound.}  The search branches on boxes of
(direction, roll) and sweeps the offset exhaustively per box.  Each box
receives an upper bound on
$\vol(T\cap(\SX\cup\Stub(c)))$ valid \emph{simultaneously for every tube in
the box}: the volume is written as an integral of line masses parallel to
the boxed axis directions (Fubini as in Lemma~\ref{lem:chord}); per-line
masses are bounded by interval-arithmetic-style slab clipping of the two
region boxes; the sweep over $(a,b)$ and $\alpha$ is dominated by a discrete
correlation with a dilated cross-section kernel; boxes provably violating
(ii) or (iii) for every contained tube are discarded (feasibility tested
with every discretization margin on the keep side); and coarse boxes fall
back on exact componentwise chord caps
($\mathrm{chord}\le\mathrm{extent}_i/|d_i|$, cf.\ Lemma~\ref{lem:chord}).
Every discretization, linearization and floating-point term is slackened in
the conservative (over-counting) direction, so the maximum of the leaf
bounds is, at \emph{every} stopping point of the search, a valid upper bound
for $V_{\mathrm{ii}}(u)$.  Runs at several $u$ produce
Table~\ref{tab:landscape}; each row, combined with
Proposition~\ref{prop:family} at its own $u$, is an independently valid
theorem for every $M>u$ (only $c<1+M$ enters, and the region
$\SX\cup\Stub(c)$ does not depend on $M$).

\begin{theorem}[computed upper bound; computer-assisted]\label{thm:C}
Let $M>2.2$ and let $(B,g)$ traverse $\K_M$ weakly.  Then
\[
\vol(B)\;\le\;3.4083 .
\]
Specifically, a $49{,}513$-leaf branch-and-bound cover of
$\mathcal A_{2.2}$ certifies $V_{\mathrm{ii}}(2.2)\le3.408293$, so
Proposition~\ref{prop:family} with $u=2.2$ gives
$\vol(B)\le\max\{2.2,\,3.408293\}<3.4083$, covering the corridor of
Problem~6.64 ($M=4$) and every $M>2.2$.
\end{theorem}

\begin{remark}[nature and rigor status of Theorem~\ref{thm:C}]
\label{rem:thmCnature}
Theorem~\ref{thm:C} is computer-assisted and its rigor status differs from
Theorems~\ref{thm:A} and~\ref{thm:B}: the computation runs in IEEE double
precision, not in interval arithmetic.  Every comparison, dilation and
per-box bound carries an explicit slack in the conservative direction
($10^{-9}$ relative and $10^{-6}$ absolute per box), the geometry enters
only through outer enclosures, and the certified value sits $\sim0.9$ below
the independent analytic ceiling $\Psi(2.2)=4.96\ldots$, so a floating-point
failure would have to defeat the slacks systematically rather than
marginally.  A re-run under full interval (or rational) arithmetic is the
natural hardening step and is planned as future work; until then
Theorems~\ref{thm:A} and~\ref{thm:B} stand as the fully hand-checkable
ladder beneath Theorem~\ref{thm:C}.
\end{remark}

\begin{table}[ht]
\centering
\caption{The computed landscape $u\mapsto V_{\mathrm{ii}}(u)$ (certified
upper enclosures, rounded up).  Each row is an independently valid bound
$\Vstar(\K_M)\le\max\{u,V_{\mathrm{ii}}(u)\}=V_{\mathrm{ii}}(u)$ for
\emph{every} $M>u$ (Proposition~\ref{prop:family}); the headline
Theorem~\ref{thm:C} is the $u=2.2$ row.  ``Budget-stopped'' rows are valid
but improvable: the frontier maximum at any stopping point is certified, and
the archived frontiers are resumable (Appendix~\ref{app:repro}).
$g$ is the evaluator's cross-plane resolution.}
\label{tab:landscape}
\small
\setlength{\tabcolsep}{4.5pt}
\begin{tabular}{ccccl}
\toprule
$u$ & certified $V_{\mathrm{ii}}(u)$ & bound for all $M>u$ & resolution & status \\
\midrule
2.0 & 4.0706 & 4.0706 & $g=0.04$ & converged at this res. \\
2.2 & \textbf{3.4083} & \textbf{3.4083} & $g=0.02$ (+ $g=0.01$ polish) & budget-stopped \\
2.4 & 3.5236 & 3.5236 & $g=0.02$ (+ polish) & budget-stopped \\
2.6 & 3.6942 & 3.6942 & $g=0.02$ (+ polish) & budget-stopped \\
3.0 & 4.2492 & 4.2492 & $g=0.04$ & budget-stopped \\
\bottomrule
\end{tabular}
\end{table}

\begin{remark}[verification gates]\label{rem:gates}
The credibility of a computed certificate rests on its test battery.  Four
gates (plus a Monte-Carlo cross-check) were run before any search and re-run
after every change to the evaluator; all pass, and their full output is
archived (Appendix~\ref{app:repro}).
\emph{G1, witness inclusion}: on the certified record traversal of
\cite{OurLB}, at the first crossings of $c=2.5/3.0/3.5$ the record's actual
tube satisfies every constraint of $\mathcal A_u$ with margin (tilt exceeds
the floor by $+0.61/+0.56/+0.48$; both window penetrations are strictly
positive with no tolerance consumed; offsets at radius $\le0.72$ of the
allowed $2.12$), and the evaluator's single-box outer value at that pose
($2.90/2.97/3.70$) dominates both the certified record volume $1.8558404$
and the pose's measured carve mass $1.867$ --- the family provably contains
the one known high-volume traversal, and the evaluator over-counts it, as it
must.
\emph{G2, exactness units}: on two analytically solvable configurations
(axis-aligned; $45^\circ$ in-plane) the computed outer bound dominates the
exact volume and converges to it at the predicted $O(g)$ rate with recorded
constants; the fast correlation matches a direct reference evaluation to
$6.3\cdot10^{-13}$; an independent Monte-Carlo cross-check on $40$ random
tubes found no case with certified value below the Monte-Carlo estimate.
\emph{G3, monotonicity}: box shrinking and grid refinement never increased a
bound ($12$ random boxes, $72$ child comparisons, zero violations).
\emph{G4, restartability}: killing and resuming the search reproduces the
frontier bit-exactly (hash equality).  A failed gate would have stopped the
computation.
\end{remark}

\begin{remark}[mechanism consistency]\label{rem:mechanism}
The computed landscape of Table~\ref{tab:landscape} has an interior minimum
in $u$, for the same reason as Theorem~\ref{thm:B}'s optimum at
$u_B\approx2.60$: raising $u$ raises the stub floor, while lowering $u$
weakens the forced-tilt floor $\mu_*(u)$, so nearly horizontal tubes may
rake ever longer stretches of the entry arm.  Pushing below $u=2.2$ indeed
makes the computed bound \emph{worse}: the $u=2.0$ run converged at its
resolution to the larger value $4.0706$.  The computed minimum sits at a
smaller $u$ than $u_B$ because the certified enclosure of
$V_{\mathrm{ii}}(u)$ is sharper than the analytic cap $\Psi(u)$
(at $u=2.2$: $3.41$ versus $4.96$), which shifts the trade-off toward the
lower stub floor.
\end{remark}

\section{Further structural facts}\label{sec:structure}

The next results are not needed for Theorems~\ref{thm:A} and~\ref{thm:B} but
constrain the problem further and, we hope, will be useful for subsequent
work.  The first quantifies the volume of the intersection of two tubes and
turns the current record lower bound into an unconditional rotation
statement.

\begin{lemma}[two-tube volume bound]\label{lem:V2}
Let $T,T'$ be tubes with unit-square cross-sections whose axis lines make
angle $\theta\in(0,\pi/2]$.  Then
\[
\vol(T\cap T')\;\le\;\frac{\sqrt2}{\sin\theta}.
\]
\end{lemma}

\begin{proof}
Let $e,e'$ be unit axis directions of $T,T'$ with $\angle(e,e')=\theta$, and
let $\bar w:=\frac{e\times e'}{|e\times e'|}$, so $\bar w\perp e,e'$ and
$|\det(e,e',\bar w)|=\sin\theta$.  Use oblique coordinates
$p=\alpha e+\beta e'+\gamma\bar w$, so $dp=\sin\theta\,d\alpha\,d\beta\,
d\gamma$.  Membership $p\in T$ depends only on the projection of $p$ onto
$e^\perp$, namely $\beta e'_\perp+\gamma\bar w+\mathrm{const}$ with
$e'_\perp:=e'-(e\cdot e')e$, $|e'_\perp|=\sin\theta$; at fixed $\gamma$ the
admissible $\beta$ fill an interval of length $\ell(\gamma)/\sin\theta$,
where $\ell(\gamma)\le\sqrt2$ is a chord of the unit square of $T$ along the
fixed direction $e'_\perp/|e'_\perp|$ at offset $\gamma$ (note
$\bar w\in e^\perp$ and $\bar w\perp e'_\perp$).  Symmetrically, membership
$p\in T'$ constrains $\alpha$ to an interval of length
$\ell'(\gamma)/\sin\theta$, where $\gamma\mapsto\ell'(\gamma)$ are parallel
chords of the unit square of $T'$ at perpendicular offset $\gamma$; by
Cavalieri, $\int_\R\ell'(\gamma)\,d\gamma=\operatorname{area}=1$.  Hence
\[
\vol(T\cap T')=\sin\theta\int_\R
\frac{\ell(\gamma)}{\sin\theta}\cdot\frac{\ell'(\gamma)}{\sin\theta}
\,d\gamma
\le\frac{\sup_\gamma\ell(\gamma)}{\sin\theta}\int_\R\ell'(\gamma)\,d\gamma
\le\frac{\sqrt2}{\sin\theta}. \qedhere
\]
\end{proof}

\begin{corollary}[rotation budget]\label{cor:RB}
Let $(B,g)$ traverse $\K_M$ weakly with $\vol(B)\ge v>\sqrt2$, and let
$a:=R_0^{-1}\xhat$, $b:=R_1^{-1}\zhat$ be the body-frame axes of the entry
and exit tubes (Lemma~\ref{lem:tube}).  Then the acute angle
$\theta_L$ between the axis \emph{lines} spanned by $a$ and $b$ satisfies
$\sin\theta_L\le\sqrt2/v$, and the net rotation $R_1R_0^{-1}$ has rotation
angle at least $90^\circ-\arcsin(\sqrt2/v)$.  In particular, at the current
certified record volume $v=1.8558404$: every weakly traversing body of at
least this volume satisfies $\theta_L\le49.65^\circ$ and must rotate in total
by more than $40.35^\circ$.
\end{corollary}

\begin{proof}
$B$ lies in both body-frame tubes $g_0^{-1}(\R\times[0,1]^2)$ and
$g_1^{-1}([0,1]\times[y_t,y_t+1]\times\R)$, with axes $a$ and $b$.  If
$\sin\theta_L>\sqrt2/v$ then Lemma~\ref{lem:V2} gives
$\vol(B)\le\sqrt2/\sin\theta_L<v$, a contradiction; so
$\sin\theta_L\le\sqrt2/v$ (this also covers $\theta_L=0$ trivially).  The
net rotation maps $\xhat=R_0a\mapsto R_1a$, and
$\angle(R_1a,\zhat)=\angle(R_1a,R_1b)=\angle(a,b)=:\theta_v\in[0^\circ,
180^\circ]$, where $\theta_L=\min(\theta_v,180^\circ-\theta_v)$.  By the
spherical triangle inequality the displacement of $\xhat$ is
$\angle(\xhat,R_1a)\ge|\,90^\circ-\theta_v\,|=90^\circ-\theta_L$, and a
rotation moves no unit vector by more than its rotation angle; hence the
rotation angle is at least $90^\circ-\theta_L\ge90^\circ-\arcsin(\sqrt2/v)$.
For $v=1.8558404$: $\arcsin(\sqrt2/v)=49.6438\ldots^\circ\le49.65^\circ$ and
$90^\circ-49.6439^\circ=40.3561\ldots^\circ>40.35^\circ$.
\end{proof}

\begin{remark}[why the events are indispensable]\label{rem:K1}
The two tube containments of Lemma~\ref{lem:tube} \emph{alone} bound
nothing: they leave the relative pose of the two tubes in the body frame
unconstrained, and if the axes coincide the intersection is an infinite tube
of infinite volume.  This degeneracy is realized, not hypothetical: on the
record traversal of \cite{OurLB} the body-frame entry and exit axes agree to
$1.2\cdot10^{-6}$ degrees (the block's long axis remains the traveling axis
through both turns).  Likewise, without the touch clauses of
Lemma~\ref{lem:event}, all containments are satisfied by a body parked deep
in the entry arm, and no volume bound can follow.  Any bound must therefore
couple the tube poses to forced events, as Theorems~\ref{thm:A}
and~\ref{thm:B} do; Corollary~\ref{cor:RB} escapes the degeneracy differently,
by \emph{assuming} a volume and concluding a rotation.
\end{remark}

\begin{lemma}[no long segments]\label{lem:seg}
Let $M>1$ and let $(B,g)$ traverse $\K_M$ weakly.  Then every straight
segment contained in $B$ has length at most $3$.
\end{lemma}

\begin{proof}
Let $\sigma\subseteq B$ be a segment with endpoints $p,q$.  The function
$t\mapsto u(t):=\max\{g_t(p)\cdot\yhat,\ g_t(q)\cdot\yhat\}$ is continuous
(maximum of two continuous functions), $u(0)\le1$ ($B_0\subseteq\SX$) and
$u(1)\ge y_t=1+M>2$ ($B_1\subseteq\SZ$).  As in Lemma~\ref{lem:event} there
is a first time $\tau$ with $u(\tau)=2$.  Write $\sigma_\tau:=g_\tau(\sigma)$,
a segment in $B_\tau\subseteq\K$ whose endpoint $y$-values are $\le2$ with
maximum exactly $2$; since $y$ is affine along a segment,
$\sigma_\tau\subseteq\K\cap\{y\le2\}=\SX\cup\Stub(2)\subseteq\{0\le z\le1\}$
by \eqref{eq:Kfacts}.  In particular the $z$-extent of $\sigma_\tau$ is at
most $1$.

Project out $z$: $\pi(x,y,z):=(x,y)$.  Then $\pi(\sigma_\tau)$ is a segment
in $\pi(\SX\cup\Stub(2))=\bigl((-\infty,0]\times[0,1]\bigr)\cup
\bigl([0,1]\times[0,2]\bigr)=:\mathrm{Sh}$, with top endpoint
$\tilde q$ at height $y=2$ and $\tilde q_x\in[0,1]$ (a point of height $2>1$
lies in $\Stub(2)$).  Points of $\pi(\sigma_\tau)$ with $y\in(1,2]$ lie in
the strip $[0,1]\times(1,2]$ (the left part of $\mathrm{Sh}$ has $y\le1$).
Let $\tilde p$ be the other endpoint.

If $\tilde p_y>1$, the whole segment lies in $[0,1]\times(1,2]$ and has
length at most $\sqrt{1+1}=\sqrt2$; with the $z$-extent,
$|\sigma|=|\sigma_\tau|\le\sqrt{2+1}<3$.

If $\tilde p_y\le1$, the segment crosses the level $y=1$ at a point $r$,
which is a limit of segment points with $y>1$, hence $r_x\in[0,1]$.  On the
subsegment from $r$ to $\tilde q$: $\Delta y=1$ and $|\Delta x|\le1$, so the
segment's direction satisfies $|dx/dy|\le1$.  On the subsegment from
$\tilde p$ to $r$: $\Delta y=1-\tilde p_y\le1$ (as $y\ge0$ on
$\mathrm{Sh}$), and the direction is the same, so $|\Delta x|\le1$ there
too.  Altogether $|\Delta x|\le2$ and $\Delta y\le2$ over
$\pi(\sigma_\tau)$, so $|\pi(\sigma_\tau)|\le2\sqrt2$, and
$|\sigma|=|\sigma_\tau|\le\sqrt{(2\sqrt2)^2+1^2}=3$.
\end{proof}

\begin{corollary}\label{cor:convex}
For $M>1$, every \emph{convex} body that weakly traverses $\K_M$ has
diameter at most $3$ and volume at most $3$.
\end{corollary}

\begin{proof}
In a convex body the segment between any two points is contained in the
body, so Lemma~\ref{lem:seg} bounds the diameter by $3$.  By
Lemma~\ref{lem:tube}, $B$ lies in a tube of unit-square cross-section; its
extent along the tube axis is at most its diameter, so
$\vol(B)\le3\cdot1=3$ by Fubini along the axis.
\end{proof}

\begin{remark}
The constant $3$ in Lemma~\ref{lem:seg} is what this argument yields:
optimizing the intermediate height $c$ in its proof leads to
$\bigl(1+\tfrac1{c-1}\bigr)^2+c^2+1$, stationary exactly at $c=2$.  We do
not claim sharpness, although $2\sqrt2$ is the classical planar
segment-passage threshold and full unit $z$-tilt through the corner is
plausible.  Note the contrast: the \emph{diameter} of a general (non-convex)
traversing body is \emph{not} bounded by $3$ --- the record body of
\cite{OurLB} has diameter $\approx3.34$; its far-apart point pairs are
joined only through the bent bulk, in accordance with Lemma~\ref{lem:seg}.
\end{remark}

\section{Discussion}\label{sec:discussion}

\subsection*{Machine verification of the constants and lemmas}

\begin{remark}[numerical stress-testing]\label{rem:witness}
Every lemma and constant above was checked numerically, in two independent
ways, before being recorded.  (1) \emph{Synthetic}: the chord and slab
estimates and the two-tube bound were tested by exact-chord quadrature and
Monte-Carlo over random tube configurations (observed maxima: $0.67$,
$0.43$, and $0.71$ of the respective proven bounds --- the last consistent
with the conjecture that the sharp two-tube constant is $1/\sin\theta$
rather than $\sqrt2/\sin\theta$, which we do not use); the extremal-point
analysis inside Lemma~\ref{lem:tiltsharp} was verified by direct maximization
over a fine grid of the constraint region, and all quoted roots, enclosures
and evaluations ($u_B$, $\Psi(u_B)$, $f(3.5217)$) were re-derived by
independent optimization.  (2) \emph{Against the record}: on the certified
$512$-pose record traversal of \cite{OurLB}, the forced-tilt inequalities of
Lemmas~\ref{lem:tilt} and~\ref{lem:tiltsharp} hold at every tested event
height with margin (measured tilts $0.68$--$1.00$ against requirements
$0.06$--$0.64$), and the chord inequalities of Lemma~\ref{lem:chord} hold at
\emph{all} $512$ poses of the path with minimum margin $\approx0.7$ in
volume.  A failed check would have refuted the corresponding lemma; none
failed.  The scripts and their JSON outputs are archived with the lower-bound
artifacts.  Theorem~\ref{thm:C}'s computation carries its own, stricter
battery (Remark~\ref{rem:gates} and Appendix~\ref{app:repro}).
\end{remark}

\subsection*{Slack inventory}

The constants are certainly far from optimal; the value of
Theorems~\ref{thm:A}--\ref{thm:C} is that the set
$\{\Vstar(K_M)\}_{M}$ is now known to be bounded, uniformly in the
``far apart'' parameter.  The slack we see, in decreasing order of expected
yield: (1) the argument uses only \emph{one} event and only the \emph{entry}
tube --- coupling the exit tube (Lemma~\ref{lem:tube}) and the trailing
(minimum-height) events at the same time, or events at two heights, should
cut both caps; (2) the computed certificate of Section~\ref{sec:certC} is
not exhausted: its archived search frontiers are resumable
(Appendix~\ref{app:repro}), its current maximum is limited by linearization
drift rather than by the tube family itself, and uncertified interior
estimates suggest that further search on the same checkpoints would certify
$\approx3.1$--$3.3$; connectivity refinements (the body lies in one
component meeting both windows) and per-box re-anchoring are further levers
there; (3) under strong traversal the analogous climb events in $z$
force a second axis component $|w\cdot\zhat|$, which strengthens every chord
denominator in Theorem~\ref{thm:B} and every direction cap in
Section~\ref{sec:certC}; (4) proving the conjectured sharp constant
$1/\sin\theta$ in Lemma~\ref{lem:V2} would improve
Corollary~\ref{cor:RB} but not the main theorems; (5) hardening
Theorem~\ref{thm:C} with interval arithmetic
(Remark~\ref{rem:thmCnature}).  On the lower-bound side,
the gap in \eqref{eq:bracket} is expected to close mostly from above.

\subsection*{Definitional forks and scope}

Upper bounds here are proved under \emph{weak} traversal, hence hold under
every stronger reading (Remark~\ref{rem:wlog}).  One caution for future
sharpening: climb events in $z$ (the analogue of Lemma~\ref{lem:event} for
the second corner) are \emph{not} forced under weak traversal --- a body may
finish inside the second corner cube without ever climbing --- so any lemma
using them must assume strong traversal.  The bounds are uniform in
$M>2.2$ (Theorem~\ref{thm:C}), $M\ge2.61$ (Theorem~\ref{thm:B}), and
$M\ge3.53$ (Theorem~\ref{thm:A}); for smaller $M$, the rows of
Table~\ref{tab:landscape} reach down to all $M>2.0$ (bound $4.0706$), below
that Remarks~\ref{rem:optA} and~\ref{rem:optB} still give (weaker,
blowing-up) bounds on $(\sqrt2,2]$, and for $M\le\sqrt2$ the method gives
nothing --- the mechanism needs stub room above height $1+\sqrt2$.

\subsection*{Rigor status of the lower bounds}

\begin{remark}[evaluator caveat]\label{rem:caveat}
The certified lower bounds of Table~\ref{tab:ladder} are produced by the
evaluation pipeline of \cite{GGTW25}, and one step of that pipeline deserves
a precise caveat: the pipeline's points-to-cubes step (a
$4\times4\times4$-sampled cube counted as inside the corridor) is asserted,
not proven, in the source notebook --- near $K$'s reflex edges a rotated cube
admits $\Theta(h)$-deep excursions between in-corridor sample points.
Point-survival soundness is proved outright in
Appendix~\ref{app:soundness} (axis-convexity of $K$ and $8$-corner exactness
for axis-aligned boxes).  This affects the \emph{absolute} rigor of all
published pipeline numbers equally (the original $1.81$ of \cite{GGTW25}
included), so comparative claims --- the entire ladder of
Table~\ref{tab:ladder} --- are untouched; and nothing in the upper-bound
results of this note rests on the rule.  The step is upgradeable: proving it
(the pipeline's tolerances may absorb the excursions) or quantifying an
$\varepsilon$-correction would upgrade every lower bound's rigor.
\end{remark}

\appendix

\section{Point-level soundness of the certified lower bounds}
\label{app:soundness}

The lower bounds in Table~\ref{tab:ladder} are computed by a pipeline that
checks a moving discretized body against the corridor at finitely many poses:
translation steps are checked pointwise at their endpoints, and rotation
steps are checked by enclosing each point's circular arc in an axis-aligned
bounding box (with exact per-coordinate arc extrema) and testing the box's
$8$ corners.  Two elementary properties of this particular corridor make
those finite checks \emph{exact} at the level of points.

\begin{lemma}[axis-convexity]\label{lem:axisconv}
The intersection of $\K_M$ with every line parallel to a coordinate axis is
a (possibly empty or unbounded) interval.
\end{lemma}

\begin{proof}
Case check on the two frozen coordinates.  Lines parallel to $\xhat$: for
$(y,z)\in[0,1]^2$ the intersection is $(-\infty,0]\cup[0,1]=(-\infty,1]$
(the $\SY$ condition $y\le y_t$ holds automatically as $y_t>1$); for
$y\in(1,y_t]$, $z\in[0,1]$ it is $[0,1]$; for $y\in[y_t,y_t+1]$, $z\ge0$ it
is $[0,1]$; otherwise empty.  Lines parallel to $\yhat$: for $x<0$ it is
$[0,1]$ if $z\in[0,1]$, else empty; for $x\in[0,1]$, $z\in[0,1]$ it is
$[0,y_t]\cup[y_t,y_t+1]=[0,y_t+1]$; for $x\in[0,1]$, $z>1$ it is
$[y_t,y_t+1]$.  Lines parallel to $\zhat$: for $x<0$, $y\in[0,1]$ it is
$[0,1]$; for $x\in[0,1]$, $y\in[0,y_t)$ it is $[0,1]$; for $x\in[0,1]$,
$y\in[y_t,y_t+1]$ it is $[0,1]\cup[0,\infty)=[0,\infty)$.  All are
intervals.
\end{proof}

Consequently, a point whose positions at both endpoints of an axis-parallel
translation lie in $\K$ stays in $\K$ throughout that translation: the
pipeline's translation check is exact, not merely sampled.

\begin{lemma}[$8$-corner exactness]\label{lem:corners}
An axis-aligned box with all $8$ corners in $\K_M$ is contained in $\K_M$.
\end{lemma}

\begin{proof}
Let the box be $[x_0,x_1]\times[y_0,y_1]\times[z_0,z_1]$ with corners in
$\K$.  If $x_0<0$: the corners with $x$-coordinate $x_0$ lie in $\SX$ (the
only piece reaching $x<0$), so $y_1\le1$ and $z_1\le1$; and any corner at
$x_1$ lies in $\K$ with $y,z\in[0,1]$, forcing $x_1\le1$; then the box lies
in $(-\infty,1]\times[0,1]\times[0,1]\subseteq\SX\cup\SY$.  If $x_0\ge0$:
all corners lie in $\SY\cup\SZ$, so $x_1\le1$.  If moreover $z_1\le1$: all
corner $y$-values lie in $[0,y_t+1]$, so every point of the box has
$x\in[0,1]$, $z\in[0,1]$, $y\in[0,y_t+1]$, hence lies in $\SY$ (if
$y\le y_t$) or in $\SZ$ (if $y\ge y_t$).  If instead $z_1>1$: the four
corners with $z$-coordinate $z_1$ force $y_0,y_1\in[y_t,y_t+1]$, so the box
lies in $[0,1]\times[y_t,y_t+1]\times[0,\infty)=\SZ$.
\end{proof}

Consequently the rotation-step check is sound: each point's arc lies in its
bounding box, whose corners are verified, so the box --- and with it the arc
--- lies in $\K$.  Together with Lemma~\ref{lem:axisconv} this proves:
\emph{every grid point that the pipeline reports as surviving has its entire
continuous trajectory inside $\K$}, so the surviving point set, moved by the
interpolated motion, genuinely stays in the corridor; this is the precise
point-level sense in which the numbers of Table~\ref{tab:ladder} are
certified, and the sense in which they are lower bounds for
$\Vstar(\K_4)$ once the points-to-cubes step (Remark~\ref{rem:caveat}) is
granted or repaired.  The traversal endpoints of the pipeline's fixed start
and end poses lie in $\SX$ and $\SZ\cap\{z\ge1\}$ respectively, so the
motions certify \emph{weak} (indeed strong) traversal in the sense of
Definition~\ref{def:traversal}.

\section{Reproducibility artifacts for Theorem~\ref{thm:C}}
\label{app:repro}

All artifacts below are publicly archived at
\begin{center}
\url{https://github.com/Lomnus-ai/sofa3d-upper-bound}
\end{center}
(paths relative to the repository root).

\begin{itemize}
\item \emph{Code}: \texttt{scripts/wu3\_outer/} --- \texttt{geometry.py}
(parameterization, frames, drift constants), \texttt{evaluator.py}
(certified-outer per-box bound: interval line clipping, dilated-kernel
correlation, feasibility masks, chord fallbacks), \texttt{bnb.py}
(branch-and-bound with atomic checkpoints), \texttt{gates.py} (the G1--G4 +
Monte-Carlo battery of Remark~\ref{rem:gates}), and the run drivers.  The
scheme itself --- the parameterization, the per-inequality soundness audit of
every over-counting direction, and the design decisions behind it --- is
presented in Section~\ref{sec:certC}; docstring references to
``derivation notes'' point to the author's working notes.
\item \emph{Gate battery output}: \texttt{reports/wu3\_gates.json}
(all gates pass; re-run after every evaluator change).
\item \emph{Results of record}: \texttt{reports/wu3\_stageA\_summary.json}
(the assembled bound and all $8$ runs: $186{,}207$ splits, $188{,}255$
leaves in total), with per-run files
\texttt{wu3\_stageA\_u$\ast$\_g0.02\_final.json},
$\ast\in\{2.2,2.4,2.6\}$, and their \texttt{\_polished\_g0.01} companions
in \texttt{reports/}.
\item \emph{Resumable search frontiers}:
\texttt{data/wu3/bnb\_u$\ast$\_g0.02\_final.npz},
$\ast\in\{2.2,2.4,2.6\}$ --- atomically
written, configuration-keyed, with bit-exact resume verified (gate G4).
Resuming them can only tighten Table~\ref{tab:landscape}: every intermediate
frontier maximum is itself a certified bound.
\item \emph{Forward guidance (uncertified)}: at $u=2.2$ the frontier is
drift-limited rather than resolution-limited ($15{,}245$ of $49{,}513$
leaves within $0.1$ of the maximum), and interior (non-certified) estimates
of the objective near the maximizing region sit around $2.8$; further search
on the same checkpoints is expected to certify $\approx3.1$--$3.3$.  Only
the certified $3.4083$ is quoted in this note.
\end{itemize}

\section*{Tool use}

The certified computation of Section~\ref{sec:certC} was carried out in Python
(NumPy/SciPy) on CPU, using the branch-and-bound driver, outer evaluator and
validation gates archived in the repository of Appendix~\ref{app:repro}; the
constants stated throughout are re-checked by the automated witness scripts
kept with the note's source.  This note was typeset in \LaTeX.

The research reported here is original work, conceived and driven by the
author.  Claude (Anthropic) was used as a writing and coding assistant: it
helped polish the exposition of this note, and assisted with parts of the
code.  All final results reported here were verified by the author, who is
responsible for the content of this note.

\end{document}